\documentclass[aps,pra,twocolumn,nofootinbib,10pt, superscriptaddress]{revtex4-1}

\usepackage{graphicx}
\usepackage{epstopdf}
\usepackage{amsmath}
\usepackage{amssymb}
\usepackage{mathrsfs}
\usepackage{amsthm}
\usepackage{bm}
\usepackage{bbm}
\usepackage{url}
\usepackage[T1]{fontenc}
\usepackage{csquotes}
\MakeOuterQuote{"}
\newtheoremstyle{note}
  {\topsep/2}               % ABOVE SPACE
  {\topsep/2}               % BELOW SPACE
  {}                      % BODY FONT
  {\parindent}            % INDENT (empty value is the same as 0pt)
  {\itshape}              % HEAD FONT
  {.}                     % HEAD PUNCTUATION
  {5pt plus 1pt minus 1pt}% HEAD SPACE
  {}

\theoremstyle{note}
\newtheorem{theorem}{Theorem}
\newtheorem{lemma}{Lemma}

\newtheorem{corollary}{Corollary}
\newtheorem{proposition}{Proposition}

\theoremstyle{definition}

\theoremstyle{remark}
\newtheorem{remark}{Remark}

\newcommand{\mrm}[1]{\mathrm{#1}}

\newcommand{\tr}{\operatorname{tr}}

\newcommand{\diag}{\operatorname{diag}}

\newcommand{\eig}{\operatorname{eig}}

\newcommand{\rep}{\mathrel{\widehat{=}}}

\newcommand{\id}{\mathbbm{1}}

\newcommand{\rme}{\mathrm{e}}
\newcommand{\rmi}{\mathrm{i}}

\newcommand{\rmr}{\mathrm{r}}

\newcommand{\rmA}{\mathrm{A}}
\newcommand{\rmB}{\mathrm{B}}
\newcommand{\rmC}{\mathrm{C}}
\newcommand{\rmD}{\mathrm{D}}

\newcommand{\rmF}{\mathrm{F}}

\newcommand{\rmL}{\mathrm{L}}

\newcommand{\rmT}{\mathrm{T}}

\newcommand{\lr}{{\mathcal{R}_\rmL}}

\newcommand{\caH}{\mathcal{H}}

\newcommand{\caN}{\mathcal{N}}
\newcommand{\caR}{\mathcal{R}}

\newcommand{\be}{\begin{equation}}
\newcommand{\ee}{\end{equation}}
\newcommand{\ba}{\begin{align}}
\newcommand{\ea}{\end{align}}

\def\<{\langle}  %% overiding the original command \<
\def\>{\rangle}  %% overiding the original command \>
\newcommand{\ket}[1]{| #1\>}

\def\outer#1#2{|#1\>\<#2|}       %% overiding the original command \outer

\newcommand{\eref}[1]{Eq.~\textup{(\ref{#1})}}

\newcommand{\esref}[1]{Eqs.~\textup{(\ref{#1})}}

\newcommand{\fref}[1]{Fig.~\ref{#1}}

\newcommand{\thref}[1]{Theorem~\ref{#1}}
\newcommand{\Thref}[1]{Theorem~\ref{#1}}
\newcommand{\thsref}[1]{Theorems~\ref{#1}}

\newcommand{\lref}[1]{Lemma~\ref{#1}}

\newcommand{\lsref}[1]{Lemmas~\ref{#1}}

\newcommand{\pref}[1]{Proposition~\ref{#1}}

\newcommand{\crref}[1]{Corollary~\ref{#1}}
\newcommand{\Crref}[1]{Corollary~\ref{#1}}

\newcommand{\cref}[1]{Conjecture~\ref{#1}}
\newcommand{\Cref}[1]{Conjecture~\ref{#1}}

\newcommand{\aref}[1]{Appendix~\ref{#1}}

\newcommand{\rcite}[1]{Ref.~\cite{#1}}
\newcommand{\rscite}[1]{Refs.~\cite{#1}}

\newcommand{\proj}[1]{| #1\rangle\!\langle #1 |}

\begin{document}

\title{Axiomatic and operational connections between the $l_1$-norm of coherence and negativity}

\author{Huangjun Zhu}
\email{Corresponding author: zhuhuangjun@fudan.edu.cn}
\affiliation{Institute for Theoretical Physics, University of Cologne,  Cologne 50937, Germany}

\affiliation{Department of Physics and Center for Field Theory and Particle Physics, Fudan University, Shanghai 200433, China}

\affiliation{Institute for Nanoelectronic Devices and Quantum Computing, Fudan University, Shanghai 200433, China}

\affiliation{State Key Laboratory of Surface Physics, Fudan University, Shanghai 200433, China}

\affiliation{Collaborative Innovation Center of Advanced Microstructures, Nanjing 210093, China}

\author{Masahito Hayashi}
\email{Corresponding author: masahito@math.nagoya-u.ac.jp}
\affiliation{Graduate School of Mathematics, Nagoya University, 
Nagoya, 464-8602, Japan}
\affiliation{Centre for Quantum Technologies, National University of Singapore, 3 Science Drive 2, 117542, Singapore}

\author{Lin Chen}
\email{Corresponding author: linchen@buaa.edu.cn}
\affiliation{School of Mathematics 
and Systems Science, Beihang University, Beijing 100191, China}
\affiliation{International Research Institute for Multidisciplinary Science, Beihang University, Beijing 100191, China}

\begin{abstract}
Quantum coherence plays a central role in various research areas. 
The $l_1$-norm of coherence is one of the most important coherence measures that are easily computable, but it is not easy to find a simple interpretation. 	We show that the $l_1$-norm of coherence is  uniquely characterized by a few simple axioms, which demonstrates
in a precise sense that it is the analog of  negativity in entanglement theory and  sum negativity in the resource theory of magic-state quantum computation. We also provide an operational interpretation of the $l_1$-norm of coherence as the maximum entanglement,  measured by the negativity, produced by  incoherent operations acting on the system and an incoherent ancilla. To achieve this goal, we clarify the relation between the  $l_1$-norm of coherence and negativity for all bipartite states, which leads to an interesting generalization of  maximally correlated states. Surprisingly, all entangled states thus obtained are distillable. Moreover, their entanglement cost and distillable entanglement  can be computed explicitly for a qubit-qudit system. 

\end{abstract}

\date{\today}
\maketitle

\section{Introduction}

Quantum coherence underlies most nonclassical phenomena characteristic of quantum physics, such as entanglement, steering, and Bell nonlocality. It is also a key resource for demonstrating quantum advantages in various information processing tasks, such as quantum metrology, quantum cryptography, and quantum computation. Recently,   coherence  has attracted unprecedented  attention fueled by the resource theoretic  formulation proposed in \rscite{Aber06,BaumCP14} and further developed in \rcite{WintY16} among many other works; see \rcite{StreAP17} for a review.

Although many coherence measures have been proposed, the $l_1$-norm of coherence stands out as one of the most important measures that are easily computable \cite{BaumCP14,StreAP17, HuHPZ17}. It is useful in studying speedup in quantum computation, such as in the
Deutsch-Jozsa algorithm  \cite{DeutJ92,Hill16} and Grover algorithm \cite{Grov96,ShiLWY17, AnanP16}.  
It features prominently in  alternative  formulations of  uncertainty relations, complementarity relations \cite{ChenH15,SingBDP15, YuanBPM17,HuHPZ17},  and wave-particle duality in  multipath interferometers \cite{PrilRM15,BeraQSP15, BagaBCH16,BiswGW17, YuanHZZ18, GaoJHQ17,HuHPZ17}. It plays a crucial role in quantifying the cohering and decohering powers of quantum operations \cite{ManiK15, BuKZW15, GarcEP16}.  In addition, the $l_1$-norm of coherence sets an upper bound for another   important coherence measure,   the robustness of coherence \cite{NapoBCP16,PianCBN16}.

Despite its prominent role in the resource theory of quantum coherence,  the $l_1$-norm of coherence has defied many attempts to find its simple interpretation. Recently, Rana et al. \cite{RanaPL16,RanaPWL17}
proved that the logarithmic $l_1$-norm of coherence is  an upper bound for the relative entropy of coherence, which is equal to the distillable coherence~\cite{WintY16}, and thereby    argued that this measure is the analog of the logarithmic negativity in entanglement theory \cite{VidaW02,AudePE03, Plen05,HoroHHH09}.
While their observation  is quite suggestive, their argument  is not so satisfactory, given that there are many different upper bounds for the relative entropy of coherence. 
It should be pointed out  that the logarithmic negativity is essentially the only useful entanglement measure that is easily computable for all states.
It sets an upper bound for the distillable entanglement  \cite{VidaW02, HoroHHH09} and a lower bound for the asymptotic exact entanglement cost under positive-partial-transpose (PPT) operations \cite{AudePE03}.

In this paper we corroborate the significance of the  $l_1$-norm of coherence by two complementary approaches. 
First, we propose several simple axioms to single out this measure among all competitors. This approach is inspired by a similar characterization of the sum negativity in the resource theory of magic-state quantum computation \cite{BravK05,VeitFGE12,MariE12,VeitMGE14, HowaC17}.
The roles of negativity in quantum computation and foundational studies are a focus of ongoing research \cite{KenfZ04,Gros06, Spek08, Ferr11,HowaWVE14,DelfABR15, PashWB15,Zhu16Q,DeBrF17}.
Our study  demonstrates,
in a precise sense, that the $l_1$-norm of coherence is the analog of  the negativity in entanglement theory and the  sum negativity in the resource theory of magic-state quantum computation, as illustrated in \fref{fig:Resource3}. This analogy is indicative of the deep connections among the three prominent resource theories.

Next, we  provide an operational interpretation of the $l_1$-norm of coherence as the maximum entanglement,  measured by the negativity, created by incoherent operations acting on the system and an incoherent ancilla. 
To this end, we show that the $l_1$-norm of coherence  is an upper bound for  the negativity and determine all states that saturate this bound,
which lead to an interesting generalization of  maximally correlated states~\cite{Rain99}. The asymptotic exact PPT entanglement cost of these states is equal to the logarithmic negativity. Moreover, such states are distillable whenever they are entangled, and
 their entanglement cost and distillable entanglement  can be computed explicitly for a qubit-qudit system.

\section{Preliminaries}

Both entanglement theory and coherence theory are special examples of resource theories \cite{HoroO13, BranG15, CoecFS16, WintY16,StreAP17}, which provide a unified framework for studying resource quantification and manipulations under restricted operations that are deemed free. 
In the case of entanglement, we are restricted to local operations and classical communication (LOCC), so only separable states can be prepared for free. 
In the case of coherence, we are restricted to incoherent operations, so only incoherent states are free. Recall that a state is incoherent if it is diagonal in the reference basis.  A quantum operation 
with Kraus representation $\{K_j\}$ is incoherent if each Kraus operator $K_j$ is incoherent in the sense that it maps every incoherent state to an incoherent state \cite{BaumCP14,WintY16,StreAP17}. 
The operation is strictly incoherent if in addition each $K_j^\dagger$ is  incoherent~\cite{WintY16}. 
For example, a unitary is  incoherent if it is monomial, that is, if each column (row)  has only one nonzero entry, in which case the unitary is also strictly incoherent.

Resource monotones are an essential ingredient of any resource theory.  A
weak monotone  does not increase under nonselective free operations, while a (strong) monotone in addition does not increase on average under  selective free operations.  Convexity is desirable but not  necessary; it is not assumed in this paper. A resource measure is a monotone that is nonnegative and usually vanishes on free states. Although many entanglement measures have been proposed, the negativity is distinguished because of its simplicity and versatility \cite{VidaW02,AudePE03, Plen05,HoroHHH09}.
The $l_1$-norm of coherence in coherence theory \cite{BaumCP14, RanaPWL17} and the sum negativity in the resource theory of magic-state quantum computation \cite{VeitMGE14} are distinguished for the same reason. All three measures have  logarithmic versions, which are additive and satisfy all basic requirements of resource measures, although they are not convex.

Consider a finite-dimensional Hilbert space $\mathcal{H}$ with a reference  orthonormal basis $\{|j\>\}$; the \emph{$l_1$-norm of coherence} \cite{BaumCP14} of a state $\rho=\sum_{jk} \rho_{jk}|j\>\<k|$ on $\mathcal{H}$ reads
\begin{equation}
C_{l_1}(\rho):=\sum_{j\neq k} |\rho_{jk}|=\sum_{j k} |\rho_{jk}|-1.
\end{equation}
It is one of the most important coherence measures that are easily computable \cite{StreAP17, HuHPZ17,Hill16,ShiLWY17, AnanP16,ChenH15,SingBDP15, YuanBPM17,PrilRM15,BeraQSP15, BagaBCH16,BiswGW17, YuanHZZ18, GaoJHQ17,ManiK15, BuKZW15, GarcEP16,NapoBCP16,PianCBN16, RanaPL16,RanaPWL17,HuF16,QiGY17}. In addition,  it sets a tight upper bound for another important coherence measure, namely, 
the robustness of coherence $C_\caR$ \cite{NapoBCP16,PianCBN16}. The two measures $C_{l_1}$ and  $C_\caR$ coincide on pure states and thus share the same convex roof.
The logarithmic $l_1$-norm of coherence $C_{\rmL}(\rho):=\log(1+C_{l_1}(\rho))$  is additive:  $C_\rmL(\rho\otimes \sigma)=C_\rmL(\rho)+C_\rmL(\sigma)$ for arbitrary states $\rho,\sigma$. In this paper log has base 2.
It is a universal  upper bound for two families of R\'enyi relative entropies of coherence \cite{ZhuHC17}, including the relative entropy of coherence $C_\rmr(\rho)$ and 
logarithmic robustness of coherence $C_\lr(\rho):=\log(1+C_\caR(\rho))$ \cite{RanaPWL17}.

Next, consider a bipartite system shared between Alice (A) and Bob (B); we assume that the reference basis  is the tensor product of respective reference bases. The  coherence measures mentioned above are defined in the same way as before. Now the $l_1$-norm of coherence  upper bounds not only the robustness of coherence $C_\caR$ but also  the robustness of entanglement $E_\caR$ \cite{HoroHHH09,ZhuHC17}. Meanwhile, the logarithmic  $l_1$-norm of coherence upper bounds the relative entropy of coherence $C_\rmr$ and logarithmic  robustness of coherence $C_\lr$ as well as the corresponding entanglement measures $E_\rmr$ and  $E_\lr$.

Here we are interested in another useful entanglement measure that has no obvious analog in coherence theory. 
The  \emph{negativity} of  $\rho$ \cite{VidaW02, Plen05,HoroHHH09} is defined as 
\begin{equation}
\mathcal{N}(\rho):=\tr|\rho^{\rmT_\rmA}|-1=\|\rho^{\rmT_\rmA}\|_1-1,
\end{equation}
where $\rho^{\rmT_\rmA}$ denotes the partial transpose of $\rho$ with respect to  subsystem A, and $\|\cdot\|_1$ is the Schatten 1-norm (or trace norm). Note that the definition in some literature differs by a factor of 2. 
For  a pure state $\rho=|\psi\>\<\psi|$ with $|\psi\>=\sum_j \sqrt{\lambda_j} |jj\>$ and $\sum_j\lambda_j=1$, the negativity reads
\begin{equation}
\label{eq:er}
\mathcal{N}(\rho)=C_{l_1}(\rho)=\sum_{j\neq k}\sqrt{\lambda_j\lambda_k}=\biggl(\sum_j \sqrt{\lambda_j}\biggr)^2-1,
\end{equation}
which is equal to the robustness of entanglement, so the two measures share the same convex roof. The logarithmic negativity $
\mathcal{N}_\rmL(\rho):=\log(1+\mathcal{N}(\rho))$
sets an upper bound for the distillable entanglement \cite{VidaW02, Plen05, HoroHHH09} and a lower bound  for 
the asymptotic exact entanglement cost under PPT operations. 
Moreover, the latter bound is tight  when $|\rho^{\rmT_\rmA}|^{\rmT_\rmA}\ge 0$ \cite{AudePE03}.

\section{Axiomatic connection between the $l_1$-norm of coherence and negativity}
\begin{figure}
	\begin{center}
		\includegraphics[width=7cm]{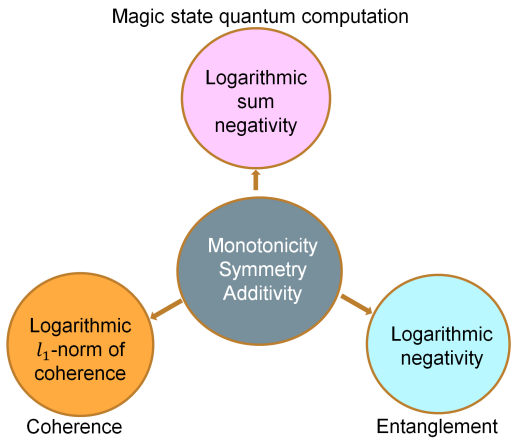}
		\caption{Common axioms that single out the $l_1$-norm of coherence, logarithmic negativity, and logarithmic sum negativity in three prominent resource theories.}
		\label{fig:Resource3}
	\end{center}
\end{figure}

Now we are ready to reveal the axiomatic connections between the $l_1$-norm of coherence and negativity as well as sum negativity; see \fref{fig:Resource3}. 
\begin{theorem}[Uniqueness of  the $l_1$-norm of coherence]\label{thm:l1Unique}
	Suppose $C(\rho)$ is a continuous weak coherence monotone that
	is a symmetric function of  nonzero 
	off-diagonal entries of $\rho$. Then $C(\rho)$ is a nondecreasing function of $C_{l_1}(\rho)$. If in addition it is additive, then $C(\rho)=c\, C_\rmL(\rho)$ for some constant $c\geq0$. 
\end{theorem}
\begin{theorem}[Uniqueness of  the negativity]\label{thm:NegUnique}
	Suppose $E(\rho)$ is a weak entanglement monotone that is a symmetric function of negative eigenvalues of $\rho^{\rmT_\rmA}$. Then $E(\rho)$ is a nondecreasing function
	of $\mathcal{N}(\rho)$. If in addition it is additive, then $E(\rho)=c\, \mathcal{N}_\rmL(\rho)$ for some constant $c\geq0$. 
\end{theorem}

Recall that a weak monotone does not increase under nonselective free operations.
A  function is symmetric if it is invariant under arbitrary permutations of its arguments. Here we implicitly assume that 
$C(\rho\otimes \sigma)=C(\rho)$ for any incoherent state $\sigma$, which  can be prepared for free; this assumption is more intuitive when $C$ is additive. Similarly, $E(\rho\otimes \sigma)=E(\rho)$ for any separable  state~$\sigma$. Remarkably, both the  $l_1$-norm of coherence and negativity can be uniquely characterized by such simple axioms. In addition, \thref{thm:l1Unique} still holds if  free operations are strictly incoherent.  
We note that a similar uniqueness theorem has been established  for the sum negativity in the resource theory of magic-state quantum computation~\cite{VeitMGE14}. Our two theorems share the same spirit with Theorem~6 in \rcite{VeitMGE14}. 
The proof of  \thref{thm:l1Unique} is considerably more involved and is relegated to \aref{app:l1UniqueProof};
 \thref{thm:NegUnique} can be proved following this proof (\lref{lem:l1Unique} in \aref{app:l1UniqueProof} in particular). Although the logarithmic sum negativity (referred to as mana in \rcite{VeitMGE14}) was not uniquely characterized  as in \thsref{thm:l1Unique} and \ref{thm:NegUnique} previously, the analogous result follows from a similar argument that underpins \thref{thm:l1Unique}.

\section{A tight inequality between the $l_1$-norm of coherence and negativity}

%\section{Pairing states}

In this section, we show that the $l_1$-norm of coherence  is an upper bound for  the negativity and determine all states that saturate this bound,
which lead to an interesting generalization of  maximally correlated states~\cite{Rain99}. The structure and entanglement properties of these states are discussed in detail. 

The following theorem  clarifies the relation between $\mathcal{N}(\rho)$ and  $C_{l_1}(\rho)$ for bipartite states; see  \aref{app:Negl1Proof} for a proof.
\begin{theorem}\label{thm:Negl1}
Any bipartite state $\rho$ on $\caH_\rmA\otimes \caH_\rmB$ satisfies $\mathcal{N}(\rho)\leq C_{l_1}(\rho)$. 
The inequality  is saturated iff $\rho^{\rmT_\rmA}$ has the form
\begin{equation}\label{eq:Negl1Equal}
\rho^{\rmT_\rmA}=\sum_{jk} a_{jk} |jk\>\<\pi(jk)|,
\end{equation}
where $\pi$ is a product of disjoint transpositions. We may assume that $\pi(jk)=jk$ whenever $a_{jk}=0$;  each transposition $(jk, j'k')$ satisfies $j\neq j'$ and $k\neq k'$; in addition, $jk'$ and $j'k$ are fixed points of $\pi$ when $(jk, j'k')$ is such a transposition. 
\end{theorem} 
Recall that  a transposition is a special permutation that interchanges two elements, while leaving others  invariant. 
 A quantum state $\rho$ is a  \emph{pairing state} if it saturates the inequality
$\mathcal{N}(\rho)\leq C_{l_1}(U\rho U^\dag)$ for some local unitary $U=U_\rmA\otimes U_\rmB$. It is a \emph{canonical pairing state} if in addition $U$ is the identity,
so that $\rho^{\rmT_\rmA}$ has the form of \eref{eq:Negl1Equal}. Here the terminology originates from the fact that every transposition is determined by a pair of basis states, and vice versa. \Thref{thm:Negl1} demands that pairing is monogamous; that is, every basis state can be paired with at most one other basis state. Moreover, some basis states must remain as bachelors.

To clarify the structure of pairing states, given a Hermitian operator $M$ acting on $\caH_\rmA\otimes \caH_\rmB$, 
define $\caN_0(M)$ as the number of negative eigenvalues of $M^{\rmT_\rmA}$; define $C_{l_0}(M)$ as the number of nonzero off-diagonal entries. The following proposition is proved in \aref{app:N0l0Proof}.
\begin{proposition}\label{pro:N0l0}
	Any state $\rho$ satisfies $C_{l_0}(\rho)\geq 2\caN_0(\rho)$.
\end{proposition}
When $\rho$ is a pairing state, $\caN_0(\rho)$ will be referred to as the \emph{pairing number}. According to \Thref{thm:Negl1} and \pref{pro:N0l0}, $\caN_0(\rho)$ is equal to the minimum of $C_{l_0}(\rho)$ under local unitary transformations. If $\rho$ is a canonical pairing state as in \eref{eq:Negl1Equal}, then $\caN_0(\rho)=C_{l_0}(\rho)/2$ is equal to  the number of transpositions in the disjoint cycle decomposition of $\pi$. A pairing state is entangled iff its pairing number is nonzero and is thus negative partial transpose (NPT). Actually, we have a  stronger conclusion, as shown in the following theorem and proved in  \aref{app:PairingDistillProof}. 
\begin{theorem}\label{thm:PairingDistill}
	Any entangled pairing state $\rho$ is NPT and distillable.	 
\end{theorem}
A family of lower bounds for the distillable entanglement of pairing states is provided in \aref{app:DistillLB}. 

Pairing states are  interesting for many other reasons.
\begin{theorem}
The asymptotic exact PPT entanglement cost $E_\mrm{PPT}(\rho)$ of any pairing state $\rho$ is equal to the logarithmic negativity, that is, $E_\mrm{PPT}(\rho)=\caN_\rmL(\rho)$. 
\end{theorem}
\begin{proof}
Without loss of generality, we may assume that $\rho$ is a canonical pairing state, which has the form of
\eref{eq:Negl1Equal}. Then $|\rho^{\rmT_\rmA}|=\sum_{jk} |a_{jk}|(|jk\>\<jk|)$ is diagonal, so that $|\rho^{\rmT_\rmA}|^{\rmT_\rmA}=|\rho^{\rmT_\rmA}|$ is positive semidefinite, which implies  
the theorem according to \rcite{AudePE03}. 
\end{proof}

\begin{proposition}
	Any pairing state $\rho$ satisfies 	$E_\rmr(\rho)\leq  E_\lr(\rho)\leq \caN_\rmL(\rho)$, where $E_\rmr(\rho)$ and $E_\lr(\rho)$ are the relative entropy of entanglement and logarithmic robustness of entanglement. 
\end{proposition}	

\begin{proof}
The first inequality  applies to all states \cite{ZhuHC17}. To prove the second one, we may assume that $\rho$ is a canonical pairing state as in \eref{eq:Negl1Equal}, so that $C_{l_1}(\rho)= \caN(\rho)$. Now the conclusion follows from the following inequality
$E_\caR(\rho)\leq C_\caR(\rho)\leq C_{l_1}(\rho)$ \cite{NapoBCP16,PianCBN16}. 
\end{proof}

Pairing states include all bipartite pure states and maximally correlated states. 
Any bipartite pure state with Schmidt rank $r$ is a pairing state with pairing number $r(r-1)/2$. A \emph{canonical maximally correlated state} \cite{Rain99, StreSDB15,WintY16} has the form
	\begin{equation}\label{eq:MC}
	\rho_{\mrm{MC}}:=\sum_{r,s=0}c_{rs}|j_r k_r\>\<j_s k_s|,
	\end{equation}
	where $j_s\neq j_r$ and $k_s\neq k_r$ whenever $s\neq r$.  Note that
	$\mathcal{N}(\rho_{\mrm{MC}})=
	C_{l_1}(\rho_{\mrm{MC}})=\sum_{r\neq s} |c
	_{rs}|$. A maximally correlated state is any state that is equivalent to $\rho_{\mrm{MC}}$
	under a local unitary transformation. Here our definition is more general than the definition in some literature \cite{Rain99,ZhuMCF17,ZhuHC17}.

In general, pairing states do not have  to be maximally correlated. Nevertheless, those with the maximal pairing number are when $\dim(\caH_\rmA)=\dim(\caH_\rmB)$. 
The following theorem is proved in \aref{app:PairingNumMaxProof}.
\begin{theorem}\label{thm:PairingNumMax}
Let $\rho$ be a pairing state on $\caH_\rmA\otimes \caH_\rmB$	with $d_\rmA:=\dim(\caH_\rmA)=\dim(\caH_\rmB)$. Then the pairing number of $\rho$ is at most $d_\rmA(d_\rmA-1)/2$, and the upper bound is saturated only if $\rho$ is  maximally correlated. 
\end{theorem}

\begin{corollary}\label{cor:Negl1Qubit}
A two-qubit state $\rho$ is a canonical pairing state iff it is diagonal or canonically  maximally correlated. 
\end{corollary}

\begin{corollary}\label{cor:Negl1QubitQudit}
A bipartite state $\rho$ on $\caH_\rmA\otimes \caH_\rmB$	with $\dim(\caH_\rmA)=2$ is a canonical pairing state iff  $\rho$  has the form
	\begin{equation}
		\label{eq:rhoQubitQudit}
		\rho=\bigoplus_{j\geq 0}^t p_j \rho_j,
	\end{equation}
where $t\leq \dim(\caH_\rmB)/2$, $p_j\geq0$, $\sum_j  p_j=1$,   $\tr_\rmA(\rho_j)$ have mutually orthogonal supports, $\rho_0$ is diagonal, and $\rho_j$ for each $j\geq 1$ is a  canonical maximally correlated state.
\end{corollary}
\Crref{cor:Negl1Qubit} is an immediate consequence of \thsref{thm:Negl1} and \ref{thm:PairingNumMax}. It is a special case of \crref{cor:Negl1QubitQudit}, which is proved in  \aref{app:Negl1QubitQuditProof}. 
Remarkably, several key entanglement (coherence) measures of the canonical pairing state in \eref{eq:rhoQubitQudit} can be computed explicitly.  Note that the distillable coherence $C_\rmD$   and coherence cost $C_\rmC$ are equal to the relative entropy of coherence $C_\rmr$ and coherence of formation $C_\rmF$, respectively \cite{WintY16}, while the
counterparts $E_\rmD$   and  $E_\rmC$
for entanglement are in general very difficult to compute. The following theorem is  proved in
\aref{app:EoFAddpairProof}.
\begin{theorem}\label{thm:EoFAddpair}
The distillable entanglement $E_\rmD$ (coherence $C_\rmD$)  and entanglement cost $E_\rmC$ (coherence cost $C_\rmC$) of the canonical pairing state $\rho$ in \eref{eq:rhoQubitQudit}  read
\begin{align}\label{eq:EntCost}
C_\rmD(\rho)=C_\rmr(\rho)=E_\rmD(\rho)=&E_\rmr(\rho)=S(\rho^{\diag})-S(\rho),\\
C_\rmC(\rho)=C_\rmF(\rho)=E_\rmC(\rho)=&E_\rmF(\rho)=\sum_j p_j E_\rmF(\rho_j),
\end{align}
where $\rho^{\diag}$ is the diagonal matrix with the same diagonal as $\rho$ and $E_\rmF(\rho_j)=H\bigl(\frac{1+\sqrt{1-\caN(\rho_j)^2}}{2}\bigr)$ with binary entropy function $H(x):=-x\log x -(1-x)\log (1-x)$.
\end{theorem}

\section{Operational connection between the $l_1$-norm of coherence and negativity}
Recently Streltsov et al. introduced a general framework for constructing coherence measures based on the maximum entanglement  produced by incoherent operations acting on the system and an incoherent ancilla \cite{StreSDB15}. More precisely, given any entanglement measure $E$, one can define a coherence measure $C_E$ as follows,
\begin{equation}\label{eq:EntGen}
C_E(\rho):=\lim_{d_\rmA\rightarrow \infty}\left\{\sup_{\Lambda_\rmi}E\left(\Lambda_\rmi\left[\rho\otimes \outer{0}{0}\right]\right)\right\}.
\end{equation}
Here $d_\rmA$ denotes the dimension of the ancilla, and the supremum is over
all operations $\Lambda_\rmi$ that are incoherent.
It turns out  that $C_E(\rho)$ coincides with the relative entropy of coherence when $E$ is the relative entropy of entanglement \cite{StreSDB15}. 
Moreover, a one-to-one mapping can be established between coherence measures and entanglement measures that are based on the convex roof~\cite{ZhuMCF17}. However, the situation is not clear for many other induced  coherence measures.

Here we are interested in the coherence measure induced from the negativity
\begin{align}\label{eq:EntGenN}
C_{\mathcal{N}}(\rho)&:=\lim_{d_\rmA\rightarrow \infty}\left\{\sup_{\Lambda_\rmi}\mathcal{N}\left(\Lambda_\rmi\left[\rho\otimes \outer{0}{0}\right]\right)\right\}
\end{align}
and the analog $C_{\mathcal{N}_\rmL}(\rho)$ from the logarithmic negativity.

\begin{theorem}\label{thm:Cohl1OI}
	\begin{align}
C_{\mathcal{N}}(\rho)&= C_{l_1}(\rho) \quad \forall \rho   \label{eq:ClCN},          \\
C_{\mathcal{N}_\rmL}(\rho)
&=C_{\rmL}(\rho)
=\log(1+C_{l_1}(\rho))
\quad \forall \rho.  \label{eq:ClCNlog}
	\end{align}
\end{theorem}
\Thref{thm:Cohl1OI} offers an operational interpretation of the $l_1$-norm  of coherence as the  maximum entanglement, measured by the 
negativity, produced by incoherent operations acting on the system and an incoherent ancilla. 
A similar interpretation applies to the logarithmic $l_1$-norm of coherence. Compared with known
connections between the $l_1$-norm of coherence and negativity, the one established here is much more precise. In addition, \eref{eq:ClCN} still holds if 
$C_{l_1}$ and $\mathcal{N}$ are replaced by their respective convex roofs $\hat{C}_{l_1}$ and $\hat{\mathcal{N}}$ according to \rcite{ZhuMCF17}, that is, $C_{\hat{\mathcal{N}}}(\rho)= \hat{C}_{l_1}(\rho)$. Similarly,  $C_{\hat{\mathcal{N}}_\rmL}(\rho)
=\hat{C}_{\rmL}$. These results further corroborate the strong connections between the $l_1$-norm  of coherence and negativity.

\begin{proof}
Note that \eref{eq:ClCNlog} is a direct consequence of \eref{eq:ClCN}, so it suffices to prove \eref{eq:ClCN}.	
The  inequality $C_\mathcal{N}(\rho)\leq C_{l_1}(\rho)$ follows from the equation
\begin{align}
&\mathcal{N}\left(\Lambda_\rmi\left[\rho\otimes \outer{0}{0}\right]\right)\leq C_{l_1}\left(\Lambda_\rmi\left[\rho\otimes \outer{0}{0}\right]\right)\nonumber\\
&\leq C_{l_1}\left(\rho\otimes \outer{0}{0}\right)
=C_{l_1}(\rho)
\end{align}
for any incoherent operation $\Lambda_\rmi$, 
where the two inequalities follow from \thref{thm:Negl1} and  the monotonicity of $C_{l_1}$, respectively.

To prove the converse inequality $C_{\mathcal{N}}(\rho)\geq C_{l_1}(\rho)$, assume that the ancilla has the same dimension as the system under consideration, so we can construct the generalized CNOT gate $U_{\mrm{CN}}|jk\>=|j (j+k)\>$, 
which  is strictly incoherent. Given $\rho=\sum_{jk}\rho_{jk}\outer{j}{k}$, the state
\begin{equation}
\tilde{\rho}_{\mrm{MC}}:= U_{\mrm{CN}}\left(\rho\otimes \outer{0}{0}\right) U_{\mrm{CN}}^\dag =\sum_{jk}\rho_{jk}\outer{jj}{kk}
\end{equation}
is canonically maximally correlated. Since $\caN(\tilde{\rho}_{\mrm{MC}})=C_{l_1}(\tilde{\rho}_{\mrm{MC}})=C_{l_1}(\rho)$,
we conclude that $C_{\mathcal{N}}(\rho)\geq C_{l_1}(\rho)$, which confirms \eref{eq:ClCN} given the opposite inequality. 
\end{proof}

\section{Summary}
We  proposed an axiomatic characterization of the
$l_1$-norm of coherence, which reveals a precise analogy to the negativity in entanglement theory and the sum negativity in the resource theory of magic-state quantum computation.  In addition,  we  provided an operational interpretation of  the $l_1$-norm of coherence as the maximum entanglement,  measured by the negativity, created by (strictly) incoherent operations acting on the system and an incoherent ancilla. Similar interpretations  apply to the  logarithmic  $l_1$-norm of coherence and its convex roof.
In the course of  study, we clarified the relation between the $l_1$-norm of coherence and negativity for  bipartite states. We also proposed pairing states as 
a generalization of   maximally correlated states and as a bridge for connecting coherence theory and entanglement theory. Surprisingly,  the asymptotic exact PPT entanglement cost of any pairing state is equal to the logarithmic negativity. Moreover, entangled pairing states are  distillable, and their
 entanglement cost and distillable entanglement can be computed explicitly for a qubit-qudit system. 

Our study  not only clarified the meaning of the $l_1$-norm of coherence, but also provided a unified perspective for understanding  three prominent resource theories, namely, entanglement, coherence, and magic-state quantum computation, which are of interest to a wide spectrum of researchers.

\bigskip

\acknowledgments
HZ acknowledges financial support from the Excellence Initiative of the German Federal and State Governments
(ZUK~81) and the DFG.
MH was supported in part by a JSPS Grant-in-Aid for Scientific Research (A) No. 17H01280, (B) No. 16KT0017, and Kayamori Foundation of Informational Science Advancement.
LC was supported by Beijing Natural Science Foundation (Grant No. 4173076), the National Natural Science Foundation (NNSF) of China (Grant No. 11501024), and the Fundamental Research Funds for the Central Universities 
(Grants No. KG12001101, No. ZG216S1760, and No. ZG226S17J6).
\bigskip

\appendix

\section*{Appendix}
In this appendix, we prove \thsref{thm:l1Unique}, \ref{thm:Negl1}, \ref{thm:PairingDistill}, \ref{thm:PairingNumMax}, \ref{thm:EoFAddpair}, \pref{pro:N0l0}, and \crref{cor:Negl1QubitQudit}  in the main text. In order to prove \thref{thm:Negl1}, 
we show that the $l_1$-norm is an upper bound
for the trace norm and determine the condition for saturating this bound. 
In addition, we  provide additional 
examples and nonexamples of pairing states and derive a family of lower bounds for the distillable entanglement of pairing states.

\section{Proof of \thref{thm:l1Unique}\label{app:l1UniqueProof}}
\Thref{thm:l1Unique} is an immediate consequence of \lsref{lem:SymUnique} and \ref{lem:l1Unique} below.

\begin{lemma}\label{lem:SymUnique}
	Suppose $C(\rho)$ is a continuous  weak coherence monotone that
	is a symmetric function of nonzero 
	off-diagonal entries of $\rho$. Then  $C(\rho)$
	is a symmetric function of the absolute values of nonzero 
	off-diagonal entries of $\rho$. 
\end{lemma}
\begin{remark}
	According to the following proof, the assumption in \lref{lem:SymUnique} actually can be relaxed as follows: $C(\rho)$ is a continuous  symmetric function of nonzero 
	off-diagonal entries of $\rho$ that is invariant under incoherent unitary transformations and satisfies $C(\rho\otimes \sigma)=C(\rho)$ whenever $\sigma$ is an incoherent state.
\end{remark}
\begin{remark}
	Given a density matrix $\rho=\sum_{jk}\rho_{jk}|j\>\<k|$, the operator $\tau(\rho):=\sum_{jk}|\rho_{jk}|(|j\>\<k|)$ is not necessarily positive semidefinite.  
	One counterexample is 
	\begin{equation}
		\rho\rep\frac{1}{4}\begin{pmatrix}
			1 &a &0&-a\\
			a&1&a&0\\
			0&a&1&a\\
			-a&0&a&1
		\end{pmatrix},\;\;
		\tau(\rho)\rep\frac{1}{4}\begin{pmatrix}
			1 &a &0&a\\
			a&1&a&0\\
			0&a&1&a\\
			a&0&a&1
		\end{pmatrix},
	\end{equation}
	where $a=\frac{1}{\sqrt{2}}$. Note that $\rho$ has eigenvalues $\frac{1}{2}, \frac{1}{2},0,0$, while $\tau$ has eigenvalues $\frac{1}{4}(1\pm\sqrt{2}),\frac{1}{4},\frac{1}{4}$. 
	
	In \lref{lem:SymUnique} we do not assume that the monotone $C$ is also defined at $\tau(\rho)$ given a density matrix $\rho$ except when $\tau(\rho)$ is positive semidefinite. In addition,  the following proof does not resort to $\tau(\rho)$, but converts $\rho$ to another density matrix. 
\end{remark}

\begin{proof}
	To prove \lref{lem:SymUnique}, we need to prove the equality $C(\tilde{\rho})=C(\rho)$ whenever the set of absolute values of off-diagonal entries of 
	$\tilde{\rho}$ is the same as that of  $\rho$ given any pair of density matrices $\rho$ and $\tilde{\rho}$. We may assume that $\rho$ and $\tilde{\rho}$ act on the $d$-dimensional Hilbert space $\caH$ with $d\geq 2$, so that 
	$C_{l_1}(\tilde{\rho})=C_{l_1}(\rho)\leq d-1$ \cite{BaumCP14,StreAP17}.

	Since $C(\rho)$ is continuous, without loss of generality, we may assume that all phase factors  $\rho_{jk}/|\rho_{jk}|$ are $L$th roots of unity for some positive integer $L$, that is, 
	\begin{equation}
		\rho_{jk}=|\rho_{jk}|\rme^{2\pi\rmi m_{jk}/L},
	\end{equation}
	where $m_{jk}$ are integers. When $\rho_{jk}=0$, the value of $\rho_{jk}/|\rho_{jk}|$ is irrelevant to the following discussion. Similarly, we assume that all $\tilde{\rho}_{jk}/|\tilde{\rho}_{jk}|$ are $\tilde{L}$th roots of unity. To prove the lemma, we shall construct a quantum state that has the same coherence measure as both $\rho$ and $\tilde{\rho}$.

	Let $K$ be the smallest integer that is divisible by $L, \tilde{L}$ and satisfies $K\geq 2d$. 
	Let $\omega=\rme^{2\pi\rmi/K}$ be a primitive $K$th root of unity. Define $G$ as  the group composed of all $d\times d$ diagonal matrices whose diagonal entries are $K$th roots of unity (that is, integer powers of $\omega$). Note that $G$ has order $K^d$. 
	Define 
	\begin{equation}
		\rho_2:=\frac{1}{K^d}U_G(\id_{K^d}\otimes \rho)U_G^\dagger,
	\end{equation}
	where $\id_{K^d}$ denotes the identity on the $K^d$-dimensional Hilbert space and 
	\begin{equation}
		U_G:=\bigoplus_{U\in G} U.
	\end{equation}
	Then $C(\rho_2)=C(\rho)$ given that $\id_{K^d}/K^d$ is an incoherent density matrix and $U_G$ is a diagonal incoherent unitary. Moreover,  $\rho_2$ has off-diagonal entry $|\rho_{jk}|\omega^a/K^d$ with multiplicity $K^{d-1}$ for each $j,k=0, 1, \ldots, d-1$ with $j\neq k$ and $a=0, 1,\ldots, K-1$. 
	
	Define
	\begin{align}
		|\psi\>&:=\frac{1}{\sqrt{K}} \sum_{a=0}^{K-1} \omega^a|a\>, \quad |\varphi\>:=\frac{1}{\sqrt{K}} \sum_{a=0}^{K-1} |a\>.  
	\end{align}
	Note that $|\psi\>\<\psi|$ has $K$ off-diagonal entries equal to $\omega^a/K$ for each $a=1,2.\ldots, K-1$. In addition, 
	$|\psi\>$ can be mapped to $|\varphi\>$ by the diagonal incoherent  unitary transformation
	\begin{equation}
		V:=\sum_{a=0}^{K-1} \omega^{-a}|a\>\<a|.
	\end{equation}
	
	In order to construct the desired density matrix, we  need to introduce the positive operator
	\begin{equation}
		M:=\frac{1}{K^{d-1}}\bigoplus_{j<k}|\rho_{jk}|(M_1\oplus M_2\oplus M_3),
	\end{equation}
	with
	\begin{equation}
		\begin{aligned}
			M_1&:= \id_{2K^{d-2}}\otimes (|\psi\>\<\psi|),\\
			M_2&:= \id_{\frac{2(K^{d-2}-1)}{K-1}}\otimes (|\varphi\>\<\varphi|),\\
			M_3&:=\frac{1}{K}\id_K\otimes(|0\>\<0|+|0\>\<1|+|1\>\<0|+|1\>\<1|).
		\end{aligned}
	\end{equation}
	When all $\rho_{jk}$ are nonzero, $M$ acts on a Hilbert space of dimension
	\begin{equation}
		d(d-1)\left(K^{d-1}+\frac{K(K^{d-2}-1)}{K-1}+K\right), 
	\end{equation}
	and it has rank
	\begin{equation}
		d(d-1)\left(K^{d-2}+\frac{K^{d-2}-1}{K-1}+\frac{K}{2}\right).
	\end{equation}
	In general, the trace of $M$ is upper bounded by 1 according to the following equation
	\begin{align}
		\tr(M)&=\left(\frac{1}{K} +\frac{K^{d-2}-1}{(K-1)K^{d-1}}+\frac{1}{K^{d-1}}\right)\sum_{jk}|\rho_{jk}|\nonumber\\
		&=\left(\frac{1}{K} +\frac{K^{d-2}-1}{(K-1)K^{d-1}}+\frac{1}{K^{d-1}}\right)C_{l_1}(\rho)\nonumber\\
		&\leq\left(\frac{1}{K}+ \frac{K^{d-2}-1}{(K-1)K^{d-1}}+\frac{1}{K^{d-1}}\right)(d-1)<1, 
	\end{align}
	where we have employed the inequality $C_{l_1}(\rho)\leq d-1$ \cite{BaumCP14,StreAP17} as well as the assumptions  $d\geq 2$ and  $K\geq 2d$.
	
	Calculation shows that the set of off-diagonal entries of $M$ is the same  as that of $\rho_2$. To verify this claim, here we assume that  all $\rho_{jk}$ with $j<k$ have different nonzero absolute values, although this assumption is not essential.
	Note that every nonzero off-diagonal entry of $\rho_2$ is equal to $|\rho_{jk}|\omega^a/K^d$
	for $a=0,1,\ldots, K-1$ and $j,k=0,1,\ldots, d-1$ with $j<k$. The same is true for $M$.
	In addition, the number of off-diagonal entries of $\rho_2$ that are equal to $|\rho_{jk}|\omega^a/K^d$ for given $j,k,a$ is  $2K^{d-1}$ (note that $|\rho_{jk}|=|\rho_{kj}|$). The counterpart for $M$ is 
	\begin{gather}
		\frac{2(K^{d-2}-1)}{K-1}\times K(K-1)+2K=2K^{d-1}
	\end{gather}
	when $a=0$ and
	\begin{gather}
		2K^{d-2}\times K=2K^{d-1}
	\end{gather}
	when  $a=1,2,\ldots, K-1$. The numbers are equal in both cases, from which our claim follows.

	Define the density matrix
	\begin{equation}
		\rho_3:=M\oplus(1-\tr M). 
	\end{equation}
	Then  the set of off-diagonal entries of $\rho_3$ is the same as that of $M$ and that of $\rho_2$. It follows from the assumption that 
	$C(\rho_3)=C(\rho_2)=C(\rho)$.

	Now we are ready to construct the desired density matrix. Let
	\begin{equation}
		\begin{aligned}
			M_1'&:=\id_{2K^{d-2}}\otimes (|\varphi\>\<\varphi|),\\
			M'&:=\frac{1}{K^{d-1}}\bigoplus_{j<k}|\rho_{jk}|(M_1'\oplus M_2\oplus M_3),\\
			\rho_4&:=M'\oplus(1-\tr M').
		\end{aligned}
	\end{equation}
	Then  $\rho_4$ and $\rho_3$ can be mapped to each other by a diagonal unitary transformation given that $|\psi\>$ and $\varphi\>$ can be mapped to each other by the diagonal unitary $V$ and that $\tr M=\tr M'$. Incidentally, each matrix element of $\rho_4$ is the absolute value of the corresponding matrix element of $\rho_3$.  Therefore,  $C(\rho_4)=C(\rho_3)=C(\rho_2)=C(\rho)$. 
	
	By the same reasoning, we have $C(\rho_4)=C(\tilde{\rho})$, given that  the set of absolute values of off-diagonal entries of 
	$\tilde{\rho}$ is the same as that of  $\rho$. 
	It follows that $C(\tilde{\rho})=C(\rho)$, which completes the proof of the lemma.	
\end{proof}

\begin{lemma}\label{lem:l1Unique}
	Suppose $C(\rho)$ is a weak coherence monotone that
	is a symmetric function of the absolute values of nonzero 
	off-diagonal entries of $\rho$. Then $C(\rho)$ is a nondecreasing function of $C_{l_1}(\rho)$. If in addition it is additive, then $C(\rho)=c\, C_\rmL(\rho)$ for some constant $c\geq0$. 
\end{lemma}
\begin{remark}
	The continuity assumption appearing in \lref{lem:SymUnique} is not required in \lref{lem:l1Unique}. 
\end{remark}

\begin{proof}
	Let  $\rho$ and $\tilde{\rho}$ be two arbitrary quantum states that satisfy  $C_{l_1}(\rho)= C_{l_1}(\tilde{\rho})$.
	If $C_{l_1}(\rho)=0$, then $\rho,\tilde{\rho}$ have no nonzero off-diagonal entries, so  $C(\rho)=C(\tilde{\rho})$ by assumption. Otherwise, suppose the absolute values of nonzero off-diagonal entries of  $\rho$ are given by $a_0, a_1, \ldots a_r$, and those of  $\tilde{\rho}$ are  given by  $b_0, b_1, \ldots b_s$.
	Let  
	\begin{equation}
		\begin{aligned}
			\sigma&\rep\frac{1}{C_{l_1}(\rho)}\diag(a_0, a_1, \ldots a_r),\\ \tilde{\sigma}&\rep\frac{1}{C_{l_1}(\rho)}\diag(b_0, b_1, \ldots b_s);
		\end{aligned}
	\end{equation}
	then $\sigma,\tilde{\sigma}$ represent incoherent quantum states.  In addition, the set of absolute values of nonzero off-diagonal entries of  $\rho\otimes \tilde{\sigma}$ is the same as that of  $\tilde{\rho}\otimes \sigma$. Therefore, $C(\rho\otimes 
	\tilde{\sigma})=C(\tilde{\rho}\otimes \sigma)$
	by assumption, which further implies that $C(\rho)=C(\tilde{\rho})$ given that $\sigma,\tilde{\sigma}$ are incoherent. This conclusion shows that $C(\rho)$ is a  function of $C_{l_1}(\rho)$ and thus also a function of $C_{\rmL}(\rho)$.
	
	Suppose $C_{l_1}(\tilde{\rho})\geq C_{l_1}(\rho)$, then we can find a state $\sigma'$ such that $C_{l_1}(\rho\otimes \sigma')=C_{l_1}(\tilde{\rho})\geq C_{l_1}(\rho)$. It follows that $C(\tilde{\rho})=C(\rho\otimes \sigma')\geq C(\rho)$, given that $C$ is a function of $C_{l_1}$ and that it is monotonic under incoherent operations. Therefore, $C(\rho)$ is a nondecreasing function of $C_{l_1}(\rho)$.

	If $C$ is also additive,  then $C(\rho^{\otimes k})=kC(\rho)$ for any $\rho$ and any positive integer $k$, which implies that $C(\rho)\geq 0$.
	If $\rho$ is incoherent, then  we have $2C(\rho)=C(\rho^{\otimes2})=C(\rho)$, which implies that $C(\rho)=0$. 
	Choose   a density matrix $\bar{\rho}$ with $C_\rmL(\bar{\rho})=1$; 
	if $C$ is not identically zero, multiply $C$ by a positive constant if necessary, we may assume that $C(\bar{\rho})=1$. Consequently $C(\bar{\rho}^{\otimes k})=k=C_\rmL(\bar{\rho}^{\otimes k})$. Therefore, $C(\rho)=C_\rmL(\rho)$ whenever $C_\rmL(\rho)$ is an integer.

	Given any positive rational number $n/m$ as the ratio of two positive integers, we can construct a quantum state $\varrho$ with $C_\rmL(\varrho)=n/m$, so that $C_\rmL(\varrho^{\otimes m})=n$. Therefore, $m C(\varrho)=C(\varrho^{\otimes m})=C_\rmL(\varrho^{\otimes m})=n$. It follows that $C(\rho)=C_\rmL(\rho)$ whenever $C_\rmL(\rho)$ is a rational number.
	Finally, the restriction to the set of rational numbers can be eliminated given that  $C(\rho)$ is nondecreasing in $C_\rmL(\rho)$ and coincides with the latter on a dense subset of the ray of nonnegative real numbers.
	Interestingly, $C$ is automatically continuous and strongly monotonic given the extra additivity assumption.
	This observation completes the proof of \lref{lem:l1Unique}. 
\end{proof}

\section{Proof of \thref{thm:Negl1}\label{app:Negl1Proof}}

\begin{proof}
	The bound $\mathcal{N}(\rho)\leq C_{l_1}(\rho)$ follows from the  equation below,
	\begin{align}
		\mathcal{N}(\rho)+1&=\tr|\rho^{\rmT_\rmA}|=\|\rho^{\rmT_\rmA}\|_1\leq \|\rho^{\rmT_\rmA}\|_{l_1}=\|\rho\|_{l_1}\nonumber\\
		&=C_{l_1}(\rho)+1,
	\end{align}
	where the inequality follows from \lref{lem:Negl1bound} in \aref{app:Tracel1}. 
	
	Suppose  $\rho^{\rmT_\rmA}$ has the form in \eref{eq:Negl1Equal} as reproduced here,
	\begin{equation}\label{eq:Negl1EqualSupp}
		\rho^{\rmT_\rmA}=\sum_{jk} a_{jk} |jk\>\<\pi(jk)|,
	\end{equation}
	where $\pi$ is a product of disjoint transpositions. In addition, $\pi(jk)=jk$ whenever $a_{jk}=0$;  each transposition $(jk, j'k')$ satisfies $j\neq j'$ and $k\neq k'$; meanwhile, $jk'$ and $j'k$ are fixed points of $\pi$ when $(jk, j'k')$ is such a transposition.  Then the inequality  $\mathcal{N}(\rho)\leq C_{l_1}(\rho)$ is saturated according to \lref{lem:Negl1bound}, which can also be verified directly.
	
	Conversely, suppose $\mathcal{N}(\rho)=C_{l_1}(\rho)$; then we have $\|\rho^{\rmT_\rmA}\|_1= \|\rho^{\rmT_\rmA}\|_{l_1}$, so that   $\rho^{\rmT_\rmA}$ has the form in \eref{eq:Negl1Equal} according to \lref{lem:Negl1bound}, where $\pi$ is a permutation. 
	If $a_{jk}=0$, then all entries of $\rho^{\rmT_\rmA}$ in the row labeled by $jk$ are zero, and so are all entries
	in the column labeled by $jk$, given that $\rho$ and $\rho^{\rmT_\rmA}$ are Hermitian. Therefore, we may assume  $\pi(jk)=jk$ without loss of generality. Let $\tau$ be the restriction of $\pi$ on those basis states with $a_{jk}\neq 0$. 
	Recall that every permutation can be written as a product of disjoint cycles. If $\tau$ has a cycle of length at least 3, then $\rho^{\rmT_\rmA}$ cannot be Hermitian. Therefore, each cycle in the disjoint cycle decomposition has length at
	most 2; in other words, $\tau$ is a product of disjoint transpositions (including the identity, which corresponds to no transposition), and the same holds for $\pi$.

	Suppose $(jk, j'k')$ is a transposition in the disjoint cycle decomposition of $\tau$. Then $a_{j'k'}=a_{jk}^*\neq0$, and  $\rho^{\rmT_\rmA}$ has the form 
	\begin{equation}
		\rho^{\rmT_\rmA}=a_{jk}|jk\>\<j'k'|+a_{jk}^*|j'k'\>\<jk|+M,
	\end{equation}
	where $M$ is a Hermitian operator whose support is orthogonal to $|jk\>$ and $|j'k'\>$. Therefore, 
	\begin{equation}
		\tr[\rho (|jk'\>\<j'k|)]=\tr[\rho^{\rmT_\rmA} (|j'k'\>\<jk|)]=a_{jk}\neq 0,
	\end{equation}
	which  implies that $\<j'k|\rho^{\rmT_\rmA}|j'k\>= \<j'k|\rho|j'k\>>0$ given that $\rho$ is positive semidefinite. Consequently,
	$a_{j'k}> 0$ and  $\tau(j'k)=j'k$; by the same token  $a_{jk
		'}> 0$ and   $\tau(jk')
	=jk'$. In particular,  $jk'$ and $j'k$ are  fixed points of $\tau$. 
	If $j=j'$ or $k=k'$, then $\tau(jk)=jk$, which contradicts the assumption that $\tau$ exchanges $jk$ and $j'k'$. Therefore, $j\neq j'$ and $k\neq k'$ whenever $(jk, j'k')$ is a transposition in the disjoint cycle decomposition of $\tau$ or $\pi$.
\end{proof}

\section{Connection between the trace norm and $l_1$-norm\label{app:Tracel1}}
Here we prove that the $l_1$-norm is an upper bound for the trace norm and determine the condition for saturating this bound. This result is needed to prove \thref{thm:Negl1}. 
\begin{lemma}\label{lem:Negl1bound}
	Every matrix $X$ satisfies $\|X\|_1\leq \|X\|_{l_1}$; the inequality is saturated iff $X$ has the following form
	\begin{equation}\label{eq:norm1}
		X=\sum_j s_j |j\>\< \pi(j)|,
	\end{equation} 
	where $s_j$ are complex numbers, and $\pi$ is a permutation.
\end{lemma}
Note that the permutation $\pi$ in \eref{eq:norm1} may be subjected to additional constraints so as to comply with the numbers of rows and columns of $X$.  
The equality condition in \eref{eq:norm1} is closely related to the singular-value decomposition of $X$. After deriving \lref{lem:Negl1bound}, we realized  that the inequality $\|X\|_1\leq \|X\|_{l_1}$ follows from  Theorem~3.32 in   \rcite{Zhan02book}, but
the equality condition has not been discussed before as far as we know.
\begin{proof}
	Our proof is based on 
	the theory of majorization \cite{MarsOA11book,Bhat97book,Niel99,ZhuMCF17}. Let $x=(x_0, x
	_1, \ldots, x_{d-1})^\rmT$ and 
	$y=(y_0, y_1, \ldots, y_{d-1})^\rmT$ be two $d$-dimensional real vectors.  The vector $x$ is
	majorized by $y$ if
	\begin{equation}
		\sum_{j=0}^k x^\downarrow_j\leq \sum_{j=0}^k y^\downarrow_j \quad \forall k=0, 1, \dots, d-1,
	\end{equation}
	and the  inequality is saturated when $k=d-1$.  Here $x^\downarrow$ is constructed from $x$ by arranging its components in decreasing order. This relation is 
	written as $x\prec y$ or $y\succ x$ henceforth.
	It can also be defined for  vectors of different dimensions, in which case we  implicitly 
	add  a number of  "0" to the vector with few components so as to match the other vector.

	Suppose $X= \sum_{j,k}X_{jk}|j\rangle \langle k|$. Define column vectors
	$u=(|X_{jk}|^2)_{jk}$, $v=\diag(X^\dag X)$, and $w=\eig(X^\dag X)$, where $\eig(X^\dag X)$ denotes the vector composed of the eigenvalues of $X^\dag X$. Then it is straightforward  to verify  that 
	\begin{equation}\label{eq:uvw}
		u\prec v\prec w,
	\end{equation}
	where the second inequality is well known in matrix analysis \cite{Bhat97book}.
	Consequently,
	\begin{align}
		\|X\|_1&=\sum_j \sqrt{w_j}
		\leq \sum_j\sqrt{v_j}\leq \sum_{jk} \sqrt{u_{jk}}\nonumber\\
		&=\sum_{jk} |X_{jk
		}|=\|X\|_{l_1}.
	\end{align}
	Here the two inequalities follow from \eref{eq:uvw} and the fact that the function $\sum_j  \sqrt{w_j}$ is Schur concave in $w$. Actually, the function is strictly Schur concave, so $w\simeq v\simeq u$ whenever  the two inequalities are saturated, which amounts to the equality $\|X\|_1=\|X\|_{l_1}$. Here $v\simeq u$ means $v\prec u$ and $u\prec v$; that is, $u$ and $v$ have the same nonzero components up to a permutation. In that case, the number of  nonzero entries of 
	$X$  is equal to the number of nonzero singular values, that is, the rank of $X$. In addition, each row of $X$ has at most one nonzero entry, and so does each column. Consequently, $X$  has the form of \eref{eq:norm1}, in which case  the inequality $\|X\|_1\leq\|X\|_{l_1}$ is indeed saturated.
\end{proof}

\subsection{Alternative proof of the inequality $\|X\|_1\leq \|X\|_{l_1}$}

\begin{proof}
	Let $X=\sum_j s_j |u_j\>\<v_j|$ be the singular-value decomposition of $X$, where $s_j>0$ are singular values of $X$, $|u_j\>$ are orthonormal, and so are $|v_j\>$. Then 
	\begin{align}
		&\|X\|_1=\sum_j s_j=\sum_j |\<u_j| X |v_j\>|=
		\sum_j \biggl|\sum_{k,l} u_{jk}^* X_{kl}v_{jl}\biggr|\nonumber\\
		&\leq \sum_{j,k,l} |X_{kl}| |u_{jk}^* v_{jl}|
		\leq \frac{1}{2}\sum_{j,k,l} |X_{kl}| \left(|u_{jk}|^2+ |v_{jl}|^2\right)\nonumber\\
		&\leq \sum_{k,l} |X_{kl}|= \|X\|_{l_1}.
	\end{align}
	Here the third inequality follows from  the fact that
	$\sum_j |u_{jk}|^2\leq 1$ and $\sum_j |v_{jl}|^2\leq 1$, given that $|u_j\>$ and  $|v_j\>$ are respectively orthonormal.
\end{proof}

\section{Proof of \pref{pro:N0l0}\label{app:N0l0Proof}}
\begin{proof}
	Note that  $C_{l_0}(\rho^{\rmT_\rmA})=C_{l_0}(\rho)$ is an even number and that the diagonal entries of $\rho^{\rmT_\rmA}$ are nonnegative. Let $r=C_{l_0}(\rho)/2$, then $\rho^{\rmT_\rmA}$ can be written as follows
	\begin{equation}
		\rho^{\rmT_\rmA}=M_0+\sum_{j=1}^r M_j,
	\end{equation}
	where $M_0$ is diagonal and positive semidefinite, while each $M_j$ for $j=1,2,\ldots, r$ is a Hermitian matrix with rank 2 and one negative eigenvalue. 
	Therefore, $\rho^{\rmT_\rmA}$ has at most $r$ negative eigenvalues, which implies the desired  inequality $C_{l_0}(\rho)\geq 2\caN_0(\rho)$. 
\end{proof}

\section{Proof of \thref{thm:PairingDistill}\label{app:PairingDistillProof}}
\begin{proof}
	Suppose $\rho$ is an entangled  pairing state, which has the form of \eref{eq:Negl1Equal}. Then $\rho$ cannot be diagonal and is thus NPT according to \thref{thm:Negl1}. 
	Suppose $(jk, j'k')$ is a transposition in the disjoint cycle decomposition of $\pi$. Then $j\ne j'$, $k\ne k'$, $a_{j'k'}=a_{jk}^*\neq0$, $a_{j'k}a_{jk'}\neq0$, and  
	\begin{align}
		\rho^{\rmT_\rmA}
		&= a_{jk}|jk\>\<j'k'|+a_{jk}^*|j'k'\>\<jk|\nonumber
		\\
		&\quad+ a_{j'k}\proj{j'k}
		+a_{jk'}\proj{jk'}
		+M,
	\end{align}
	where $M$ is a Hermitian operator whose support is orthogonal to the four kets $|jk\>$, $|jk'\>$, $|j'k\>$, and $|j'k'\>$; cf.~the proof of \thref{thm:Negl1}. Let 
	\begin{equation}
P=(\proj{j}+\proj{j'})\otimes(\proj{k}+\proj{k'})
	\end{equation}
be a local projector. Then 
	\begin{align}
		&(P\rho P)^{\rmT_\rmA}=P\rho^{\rmT_\rmA}P=a_{jk}|jk\>\<j'k'|+a_{jk}^*|j'k'\>\<jk|\nonumber
		\\
		&\quad+ a_{j'k}\proj{j'k}
		+a_{jk'}\proj{jk'},
	\end{align}
	so 	$P\rho P$ is  a subnormalized two-qubit entangled  maximally correlated state, which is distillable according to \rscite{Rain99,HoroHH97}. 
	It follows that  any entangled pairing state is both NPT and distillable. 
\end{proof}

\section{Proof of \thref{thm:PairingNumMax}\label{app:PairingNumMaxProof}}
\begin{proof}
	Note that $\rho$ and  $\rho^{\rmT_\rmA}$ have the same numbers of nonzero diagonal entries and  off-diagonal entries. 
	Suppose $\rho$ is a canonical pairing state with pairing number $m$ (without loss of generality).  Then both $\rho$ and $\rho^{\rmT_\rmA}$ have $2m$ nonzero off-diagonal entries; meanwhile, the number of nonzero entries is at most $d_\rmA^2$ according to \thref{thm:Negl1}. Suppose $\rho$ has $n$ nonzero diagonal entries. 
	Then $2m\leq n(n-1)$ and $2m+n\leq d_\rmA^2$, which imply that $m\leq d_\rmA(d_\rmA-1)/2$. If the inequality is saturated, then $n=d_\rmA$, and $\rho$ has the following form
	\begin{equation}
		\rho=\sum_{r,s=0}^{d_\rmA-1}c_{rs}|j_r k_r\>\<j_s k_s|,
	\end{equation}
	where $c_{rs}\neq 0$ for all $r,s$, and $(j_s, k_s)\neq(j_r, k_r)$ whenever $s\neq r$. Now \thref{thm:Negl1} further implies  that $j_s\neq j_r$ and $k_s\neq k_r$ whenever $s\neq r$, so that $\rho$ is a canonical maximally correlated state.
\end{proof}

\section{Proof of \crref{cor:Negl1QubitQudit}\label{app:Negl1QubitQuditProof}}

\begin{proof}
	If $\rho$ is given  in \eref{eq:rhoQubitQudit}, then $\rho^{\rmT_\rmA}=\bigoplus_{j\geq 0}^t p_j \rho_j^{\rmT_\rmA}$, so that $\caN(\rho)=\sum_j p_j\caN(\rho_j)=\sum_j p_j C_{l_1}(\rho_j)=C_{l_1}(\rho)$. 
	
	If $\rho$ is a canonical pairing state,  which saturates the inequality $\mathcal{N}(\rho)= C_{l_1}(\rho)$, then  $\rho$ can be written as follows  according to \thref{thm:Negl1},
	\begin{align}
		&\rho^{\rmT_\rmA}=\sum_{s}\bigl( a_s |0j_s\>\< 0 j_s|+ b_s|0k_s\>\< 1 j_s|+ b_s^* |1 j_s\>\< 0 k_s|	\nonumber\\
		&+c_s  |1 k_s\>\< 1 k_s|\bigr)+\sum_{r}\bigl(a_r' |0 j_r'\>\<0j_r'|+c_r'|1 k_r'\>\<1 k_r'|\bigr),
	\end{align}	
	where all $j_s, k_s, j_r', k_r'$ are distinct. Therefore, 	$\rho$ has the form of \eref{eq:rhoQubitQudit} with constraints as specified.
\end{proof}	
\Crref{cor:Negl1QubitQudit} fails when $\dim(\caH_\rmA)\geq3$. One counterexample is
\begin{equation}
	\rho=\frac{1}{4}(\proj{02}+\proj{20}+\proj{\psi}+\proj{\varphi}),
\end{equation}
where $\ket{\psi}=(\ket{00}+\ket{11})/\sqrt{2}$ and $\ket{\varphi}=(\ket{11}+\ket{22})/\sqrt{2}$.

\section{Proof of \thref{thm:EoFAddpair}\label{app:EoFAddpairProof}}
Suppose $\rho$ is a canonical pairing state on $\caH_\rmA\otimes \caH_\rmB$, which has the form of \eref{eq:rhoQubitQudit}.  Let $n$ be an arbitrary positive integer and $\tilde{\rho}$  an arbitrary bipartite state shared by Alice and Bob. Then \thref{thm:EoFAddpair} follows from \esref{eq:REEpair} to \eqref{eq:EoFpairAddG} below.
\begin{gather}
	E_\rmr(\rho)=C_\rmr(\rho)=S(\rho^{\diag})-S(\rho), \label{eq:REEpair} \\
	E_\rmr^\infty(\rho)=\frac{1}{n}E_\rmr(\rho^{\otimes n})=E_\rmr(\rho), \label{eq:REEpairAdd} \\
	C_\rmD(\rho)=E_\rmD(\rho)=E_\rmr(\rho), \label{eq:DEpair}\\
	C_\rmF(\rho)=E_\rmF(\rho)=\sum_{j} p_j E_\rmF(\rho_j)\nonumber\\
	=\sum_{j} p_j H\Bigl(\frac{1+\sqrt{1-\caN(\rho_j)^2}}{2}\Bigr), \label{eq:EoFpair}\\
	C_\rmC(\rho)=E_\rmC(\rho)=\frac{1}{n}E_\rmF(\rho^{\otimes n})=E_\rmF(\rho), \label{eq:EoFpairAdd} \\
	E_\rmF(\rho\otimes \tilde{\rho})=E_\rmF(\rho)+E_\rmF(\tilde{\rho}). \label{eq:EoFpairAddG}
\end{gather}
Here $\rho^{\diag}$ is the diagonal matrix with the same diagonal as $\rho$. Note that the formula $C_\rmr(\rho)=S(\rho^{\diag})-S(\rho)$ applies to all quantum states \cite{BaumCP14}. In addition,
\begin{equation}
	E_\rmr^\infty(\rho):=\lim_{k\to\infty}\frac{1}{k}E_\rmr(\rho^{\otimes k}),\quad E_\rmC(\rho):=\lim_{k\to\infty}\frac{1}{k}E_\rmF(\rho^{\otimes k}),
\end{equation}
where $E_\rmr^\infty(\rho)$ is known as the regularized relative entropy of entanglement. 

If the parameter $p_0$ characterizing  $\rho$ in \eref{eq:rhoQubitQudit} vanishes, then $\rho$ is a direct sum of maximally correlated states, so that $S(\rho^{\diag})=S(\rho_\rmB)$,
where $\rho_\rmB=\tr_\rmA(\rho)$. 
Consequently,
\begin{equation}\label{eq:PairingDisREE}
	C_\rmD(\rho)=C_\rmr(\rho)=E_\rmD(\rho)=E_\rmr(\rho)=S(\rho_\rmB)-S(\rho).
\end{equation}
This formula has the same form as that for a canonical maximally correlated state. Recall that any canonical maximally correlated state $\sigma$ satisfies \cite{Rain99,HoroSS03,DeveW05,WintY16,ZhuHC17}
\begin{align}
	C_\rmD(\sigma)=&C_\rmr(\sigma)=E_\rmD(\sigma)=E_\rmr(\sigma)=S(\sigma_\rmB)-S(\sigma),  \label{eq:REEMC}\\
	C_\rmC(\sigma)=&C_\rmF(\sigma)=E_\rmC(\sigma)=E_\rmF(\sigma).
\end{align}

Before proving \esref{eq:REEpair} to \eqref{eq:EoFpairAddG}, we recall that the three entanglement measures $E_\rmr, E_\rmF, E_\rmD$ do not increase on average under selective LOCC. Meanwhile, $E_\rmr, E_\rmF, E_\rmC$ are convex \cite{Horo01M, DonaHR02}. By contrast, the coherence measures $C_\rmD=C_\rmr$ and $C_\rmC=C_\rmF$ are convex and additive, and do not increase on average under selective incoherent operations \cite{Aber06,BaumCP14,WintY16}.
All eight entanglement and coherence measures mentioned above are additive for maximally correlated states.

\begin{proof}[Proof of \eref{eq:REEpair}]
	Suppose the density matrix $\rho$ has the form of  \eref{eq:rhoQubitQudit}. Denote by 	$\tilde{P}_j$ for $j>0$ the projector onto   the support of $\tr_\rmA (\rho_j)$, which  has rank at most 2. Let  $\tilde{P}_0=\id_\rmB-\sum_{j>0} \tilde{P}_j$, where $\id_\rmB$ denotes the identity on $\caH_\rmB$.  Let $P_j=\id_\rmA\otimes \tilde{P}_j$ for $j\geq 0$. Then the completely positive trace-preserving (CPTP) map defined by the set of Kraus operators $\{P_j\}$ is local and also incoherent. This map turns  $\rho$ into $\rho_j$ with probability $p_j$. Therefore, 
	\begin{align}\label{eq:REEselMonPair}
		E_\rmr(\rho)\geq \sum_{j} p_j E_\rmr(\rho_j), 
	\end{align}	
	because $E_\rmr$ does not increase on average under LOCC. On the other hand, the convexity of $E_\rmr$ implies that
	\begin{align}\label{eq:REEconvPair}
		E_\rmr(\rho)\leq \sum_{j} p_j E_\rmr(\rho_j). 
	\end{align}	
	The above reasoning still applies if $E_\rmr$ is replaced by $C_\rmr$. 		In conjunction with \eref{eq:REEMC}, we deduce that
	\begin{align}\label{eq:REEpairProof}
		E_\rmr(\rho)&= \sum_{j} p_j E_\rmr(\rho_j)= \sum_{j} p_j C_\rmr(\rho_j)=C_\rmr(\rho)\nonumber\\
		&=S(\rho^{\diag})-S(\rho),
	\end{align}
	where the last equality applies to all quantum states  \cite{BaumCP14,StreAP17}. This conclusion confirms \eref{eq:REEpair}. 
\end{proof}

\begin{proof}[Proof of \eref{eq:REEpairAdd}]
	To prove \eref{eq:REEpairAdd}, it suffices to prove the equality $E_\rmr(\rho^{\otimes n})=nE_\rmr(\rho)$. 
	To 	simplify the notation, here we illustrate the argument in the case $n=2$, which admits straightforward generalization.  
	\begin{align}
		&E_\rmr(\rho^{\otimes 2})=E_\rmr\biggl(\sum_{j_1, j_2} p_{j_1}p_{j_2} \rho_{j_1}\otimes \rho_{j_2}\biggr)\nonumber\\
		&= \sum_{j_1, j_2} p_{j_1}p_{j_2} E_\rmr(\rho_{j_1}\otimes \rho_{j_2})\nonumber\\
		&=\sum_{j_1, j_2} p_{j_1}p_{j_2}[ E_\rmr(\rho_{j_1})+E_\rmr(\rho_{j_2})]
		=2E_\rmr(\rho).
	\end{align}
	Here the second equality follows from  a similar reasoning that leads to \eref{eq:REEpairProof}. The third equality follows from the fact that $\rho_j$ is either maximally correlated or separable and that the relative entropy of entanglement is additive for maximally correlated states \cite{Rain99, ZhuHC17}.  
\end{proof}

\begin{proof}[Proof of \eref{eq:DEpair}]
	By the same reasoning that leads to \eref{eq:REEselMonPair}, we have
	\begin{align}
		E_\rmD(\rho)\geq \sum_j p_j E_\rmD(\rho_j)=\sum_j p_j E_\rmr(\rho_j)=E_\rmr(\rho), 
	\end{align}	
	note that $E_\rmD=E_\rmr$	for maximally correlated states \cite{DeveW05} and that $E_\rmr(\rho)=\sum_j p_j E_\rmr(\rho_j)$ according to \eref{eq:REEpairProof}. Since the opposite inequality $E_\rmD(\rho)\leq E_\rmr(\rho)$ holds in general, we conclude that $E_\rmD(\rho)=E_\rmr(\rho)$. 
	Similar reasoning implies that
	$C_\rmD(\rho)=C_\rmr(\rho)$, which confirms \eref{eq:DEpair} given \eref{eq:REEpair}. Incidentally, the equality $C_\rmD(\rho)=C_\rmr(\rho)$ holds for all quantum states according to \rcite{WintY16}. 
\end{proof}

\begin{proof}[Proof of \eref{eq:EoFpair}] 	According to a similar reasoning that leads to \eref{eq:REEpairProof},
	\begin{align}
		C_\rmF(\rho)&=E_\rmF(\rho)= \sum_{j} p_j E_\rmF(\rho_j)\nonumber\\
		&= \sum_j p_j H\Bigl(\frac{1+\sqrt{1-\caN(\rho_j)^2}}{2}\Bigr).
	\end{align} 
	Here the last equality follows from   Wootters' formula for the entanglement of formation of  two-qubit states \cite{Woot98} and the fact that the concurrence is equal to the negativity for any two-qubit maximally correlated state and separable state. 	
\end{proof}

\begin{proof}[Proof of \eref{eq:EoFpairAdd}] This equation follows from a similar argument that leads to \eref{eq:REEpairAdd} and is also a corollary of \eref{eq:EoFpairAddG} proved below. 
	Note that $C_\rmC=C_\rmF$ is additive in general \cite{WintY16}.  
\end{proof}

\begin{proof}[Proof of \eref{eq:EoFpairAddG}]
	According to a similar reasoning that leads to \eref{eq:REEpairProof}, we have
	\begin{align}
		E_\rmF(\rho\otimes \tilde{\rho})&= \sum_{j} p_j E_\rmF(\rho_j\otimes \tilde{\rho})= \sum_{j} p_j [E_\rmF(\rho_j)+E_\rmF(\tilde{\rho})]\nonumber\\
		&=
		E_\rmF(\rho)+E_\rmF(\tilde{\rho}),
	\end{align}	
	where the second equality follows from 	 \rscite{VidaDC02,HoroSS03} since each $\rho_j$ is either maximally correlated or separable. 
\end{proof}

\section{More examples and nonexamples of pairing states}
In this section we provide additional examples and nonexamples of pairing states by considering mixtures of a bipartite entangled pure state and the completely mixed state. 
\begin{proposition}
	Let $|\psi\>$ be a bipartite entangled pure state in $\caH_\rmA\otimes \caH_\rmB$ and $\rho=p|\psi\>\<\psi|+(1-p)/(d_\rmA d_\rmB)$ with $0\leq p\leq 1$. Then $\rho$ is a pairing state iff $p=0$ or $p=1$.
\end{proposition}
\begin{proof}
	When $p=0$, $\rho$ is the completely mixed state and a separable pairing state.  When $p=1$, $\rho$ is a bipartite pure state, which saturates the inequality $C_{l_1}(\rho)\leq \caN(\rho)$ in the Schmidt basis, so  $\rho$ is also a pairing state.
	
	When $0<p<1$, we have
	\begin{equation}
		\caN(\rho)<p\caN(|\psi\>\<\psi|)\leq pC_{l_1}(|\psi\>\<\psi|)=C_{l_1}(\rho). 
	\end{equation}
	Therefore, $\rho$ cannot be a pairing state. 
\end{proof}

\section{Lower bounds for the distillable entanglement of pairing states\label{app:DistillLB}}
Here we provide a family of lower bounds for the distillable entanglement of canonical pairing states.  

Suppose $\rho$ is a canonical  pairing state that has the form of \eref{eq:Negl1EqualSupp}. 
Let $
A_1, A_2,\ldots, A_k$ be disjoint two-element subsets of $\{0,1,\ldots, d_\rmA-1\}$ and  
\begin{equation}
	P_j:=\biggl(\sum_{m\in A_j} |m\>\<m|\biggr)\otimes \id_\rmB,\quad j=1, 2, \ldots, k.
\end{equation}
Then $P_1, P_2, \ldots, P_k$ define a (possibly incomplete) local projective measurement. 
Let 
\begin{equation}
	\tilde{\rho}_j :=\frac{P_j\rho P_j}{p_j},\quad p_j=\tr(P_j\rho P_j). 
\end{equation}
Then each $\tilde{\rho}_j$ with $p_j>0$ is a $2\times d_\rmB$ pairing state. 
If  $\rho$ is   entangled, then it is possible to choose at least one subset, say $A_1$, such that $\tilde{\rho}_1$ is entangled.

Suppose  $
A_1, A_2, \ldots, A_k$ are chosen such that all states $\tilde{\rho}_1, \tilde{\rho}_2 \ldots, \tilde{\rho}_k$ are entangled.
Then  the distillable entanglement of $\rho$ can be lower bounded as follows,
\begin{align}
	E_\rmD(\rho)&\geq \sum_{j=1}^k p_j E_\rmD(\tilde{\rho}_j)=\sum_{j=1}^k p_j E_\rmr(\tilde{\rho}_j)\nonumber\\
	&=\sum_{j=1}^k p_j \bigl[S(\tilde{\rho}_j^{\diag})-S(\tilde{\rho}_j)\bigr],
\end{align}
where the inequality follows from the monotonicity of $E_\rmD$ under LOCC, and the two equalities follow from 
\esref{eq:DEpair} and \eqref{eq:REEpair}, respectively.

\nocite{apsrev41Control}
\bibliographystyle{apsrev4-1}
\bibliography{all_references}

%@CONTROL{REVTEX41Control}
%@CONTROL{apsrev41Control,author="48",editor="1",pages="1",title="0",year="0"}

\end{document}